\newcommand*{\ARXIV}{}
\newcommand*{\JSA}{}

\ifdefined\JSA
\documentclass[]{interact}   
\fi
\ifdefined\JQAS
\documentclass[letterpaper,12pt]{article}   
\fi

\usepackage{amsmath,amssymb,amscd,latexsym,dsfont,wasysym,bbm}
\usepackage{float}
\usepackage{color}
\usepackage{graphics}
\usepackage{graphicx}
\usepackage{subfigure}
\usepackage{multirow,multicol} 
\usepackage{verbatim}
\usepackage{psfrag}
\usepackage{comment}
\usepackage{makeidx}
\usepackage[]{algorithm}
\usepackage{algorithmicx}
\usepackage{algpseudocode}
\usepackage{cases}
 \usepackage{setspace}
\usepackage{pgfplots} 
\newcommand{\tr}[1]{\textrm{#1}}
\newcommand{\mr}[1]{\mathrm{#1}}
\newcommand{\tnr}[1]{{\textnormal{#1}}}

\newcommand{\mc}[1]{\mathcal{#1}}

\newcommand{\ms}[1]{\mathds{#1}}

\newcommand{\ov}[1]{\overline{#1}}

\newcommand{\un}[1]{\underline{#1}}

\newcommand{\brho}{\boldsymbol{\rho}}

\newcommand{\secref}[1]{Sec.~\ref{#1}}

\newcommand{\exref}[1]{Example~\ref{#1}}

\newcommand{\lemref}[1]{Lemma~\ref{#1}}

\newcommand{\tabref}[1]{Table~\ref{#1}}

\newcommand{\ie}{i.e.,~} 		
\newcommand{\eg}{e.g.,~}	

\newcommand{\argmax}{\mathop{\mr{argmax}}}

\newcommand{\set}[1]{\{#1\}}
\newcommand{\SET}[1]{\left\{#1\right\}}

\newcommand{\cd}{\cdot}
\newcommand{\ld}{\ldots}

\newcommand{\e}{\mr{e}}

\newcommand{\PR}[1]{\Pr\SET{#1}}       	
\newcommand{\pdf}{f}            			

\newcommand{\IND}[1]{\ms{I}\big[{#1}\big]}   	
\newcommand{\Ex}{\ms{E}}     			

\newcommand{\mcI}{\mc{I}}

\newcommand{\matD}{\tnr{\textbf{D}}}

\newcommand{\matW}{\tnr{\textbf{W}}}

\ifdefined\JQAS
\usepackage{dgjournal}   
\fi

\usepackage{amsthm}

\pgfplotsset{compat=1.12}
\usetikzlibrary{positioning,shadows,shapes,fit,arrows,backgrounds}

\tikzstyle{rect_my} = [draw, rectangle, minimum width=2cm, text width=1.8cm, fill=gray!15, 
  text centered,  minimum height=.9cm]
\tikzstyle{square_my} = [draw, rectangle, minimum width=1cm, text width=0.8cm, fill=gray!15, 
  text centered,  minimum height=.9cm]
\tikzstyle{square_my_graph} = [draw, rectangle, minimum width=1.2cm, text width=1cm, fill=gray!15, 
  text centered,  minimum height=1.2cm]
\tikzstyle{circle_my} = [draw, circle, minimum width=1cm, text width=0.8cm, fill=gray!15, 
  text centered,  minimum height=.9cm]
\tikzstyle{circle_my_graph} = [draw, circle, minimum width=1.1cm, text width=.8cm, fill=gray!15, 
  text centered]
\tikzstyle{cloud_my} = [draw, shape=cloud, minimum width=1cm, text width=0.8cm, fill=gray!15, 
  text centered,  minimum height=1cm]

\tikzstyle{point_my} = [draw=none, minimum width=0cm, text width=0cm, fill=none, 
  text centered,  minimum height=0cm]    
\tikzstyle{line_my} = [draw, -latex]    
\tikzstyle{box_my}=[draw, minimum size=2em, text width=4.5em, text centered]
\tikzstyle{bigbox_my}=[draw, inner sep=15pt]
\tikzstyle{arrow_my} = [thick,->,>=stealth]
\tikzstyle{noarrow_my} = [thick,-,=>stealth]

\usepackage[acronym,nonumberlist]{glossaries} 
\usepackage{url}

\ifdefined\JSA      
\usepackage[natbibapa,nodoi]{apacite}
\fi

\newacronym[\glsshortpluralkey=PDFs,\glslongpluralkey=probability density functions]{pdf}{PDF}{probability density function}
\newacronym[\glsshortpluralkey=CDFs,\glslongpluralkey=cumulative density functions]{cdf}{CDF}{cumulative density function}
\newacronym[\glsshortpluralkey=CCDFs,\glslongpluralkey=complementary cumulative density functions]{ccdf}{CDF}{complementary cumulative density function}
\newacronym[\glsshortpluralkey=PMFs,\glslongpluralkey=probability mass functions]{pmf}{PMF}{probability mass function}
\newacronym[]{lhs}{l.h.s.}{left-hand side}
\newacronym[]{rhs}{r.h.s.}{right-hand side} 

\newacronym[]{bicm}{BICM}{bit-interleaved coded modulation}
\newacronym[]{bicmid}{BICM-ID}{BICM with iterative demapping}
\newacronym[]{cm}{CM}{coded modulation}
\newacronym[]{tcm}{TCM}{trellis-coded modulation}
\newacronym[]{mlc}{MLC}{multi-level coding}
\newacronym[]{pam}{PAM}{pulse amplitude modulation}
\newacronym[]{bpsk}{BPSK}{binary phase shift keying}
\newacronym[]{qam}{QAM}{quadrature amplitude modulation}
\newacronym[]{16qam}{16-QAM}{16-points quadrature amplitude modulation}
\newacronym[]{psk}{PSK}{phase shift keying}
\newacronym[\glsshortpluralkey=LLRs,\glslongpluralkey=logarithmic likelihood ratios]{llr}{LLR}{logarithmic likelihood ratio}
\newacronym[]{oc}{OC}{operating characteristic}
\newacronym[]{dmp}{DMP}{discretized message passing}
\newacronym[]{mp}{MP}{message passing}
\newacronym[]{ep}{EP}{expectation propagation}
\newacronym[\glsshortpluralkey=MIs,\glslongpluralkey=mutual informations]{mi}{MI}{mutual information}
\newacronym[\glsshortpluralkey=GMIs,\glslongpluralkey=generalized mutual informations]{gmi}{GMI}{generalized mutual information}
\newacronym[]{eesm}{EESM}{exponential effective-SNR-mapping}
\newacronym[]{bicm-gmi}{BICM-GMI}{BICM generalized mutual information}
\newacronym[]{awgn}{AWGN}{additive white Gaussian noise}
\newacronym[]{bsc}{BSC}{binary symetric channel}
\newacronym[]{amc}{AMC}{adaptive modulation and coding}
\newacronym[]{csi}{CSI}{channel state information}
\newacronym[]{cqi}{CQI}{channel quality indicator}
\newacronym[]{kl}{KL}{Kullback-Leibler}
\newacronym[]{cmm}{CMM}{circular moment matching}
\newacronym[]{ga}{GA}{Gaussian approximation}

\newacronym[]{sp}{SP}{set-partitioning}
\newacronym[]{gsm}{GSM}{global system for mobile communications}
\newacronym[]{edge}{EDGE}{enhanced data rates for GSM evolution}
\newacronym[]{3gpp}{3GPP}{3rd generation partnership project}
\newacronym[]{umts}{UMTS}{Universal Mobile Telecommunication System}
\newacronym[]{lte}{LTE}{Long Term Evolution}
\newacronym[]{dvb}{DVB}{digital video broadcasting}
\newacronym[]{fdd}{FDD}{Frequency Division Duplexing}

\newacronym[\glsshortpluralkey=CCs,\glslongpluralkey=convolutional codes]{cc}{CC}{convolutional code}
\newacronym[\glsshortpluralkey=PCCCs,\glslongpluralkey=parallel concatenated convolutional codes]{pccc}{PCCC}{parallel concatenated convolutional code}
\newacronym[\glsshortpluralkey=TCs,\glslongpluralkey=turbo codes]{tc}{TC}{turbo code}
\newacronym{ldpc}{LDPC}{low-density parity-check}
\newacronym[]{ofdm}{OFDM}{orthogonal frequency-division multiplexing}
\newacronym[]{bep}{BEP}{bit-error probability}
\newacronym[]{wep}{WEP}{word-error probability}
\newacronym[]{sep}{SEP}{symbol-error probability}
\newacronym[]{pep}{PEP}{pairwise-error probability}
\newacronym[]{ttcm}{TTCM}{turbo-trellis coded modulation}
\newacronym[]{uep}{UEP}{unequal error protection}
\newacronym[\glsshortpluralkey=CENCs,\glslongpluralkey=convolutional encoders]{cenc}{CENC}{convolutional encoder}
\newacronym[]{mimo}{MIMO}{multiple-input multiple-output}
\newacronym[\glsshortpluralkey=SNRs,\glslongpluralkey=signal-to-noise ratios]{snr}{SNR}{signal-to-noise ratio}
\newacronym[\glsshortpluralkey=SINRs,\glslongpluralkey=the signal-to-interference-plus-noise ratios]{sinr}{SINR}{the signal-to-interference-plus-noise ratio}
\newacronym[]{msb}{MSB}{most-significative bit}
\newacronym[]{bcjr}{BCJR}{Bahl--Cocke--Jelinek--Raviv}
\newacronym[]{cbc}{CBC}{Colavolpe--Barbieri--Caire}
\newacronym[]{skr}{SKR}{Shayovitz--Kreimer--Raphaeli}
\newacronym[\glsshortpluralkey=SEDs,\glslongpluralkey=squared Euclidean distances]{sed}{SED}{squared Euclidean distance}
\newacronym[\glsshortpluralkey=EDs,\glslongpluralkey=Euclidean distances]{ed}{ED}{Euclidean distance}
\newacronym[\glsshortpluralkey=MEDs,\glslongpluralkey=minimum Euclidean distances]{med}{MED}{minimum Euclidean distance}
\newacronym[]{core}{CoRe}{constellation rearrangement}
\newacronym[]{pdl}{PDL}{parallel decoding of the individual levels}
\newacronym[\glsshortpluralkey=GCs,\glslongpluralkey=Gray codes]{gc}{GC}{Gray code}
\newacronym[]{brgc}{BRGC}{binary-reflected Gray code}
\newacronym[]{nbc}{NBC}{natural binary code}
\newacronym[]{fbc}{FBC}{folded-binary code}
\newacronym[]{bsgc}{BSGC}{binary semi-Gray code}
\newacronym[]{msp}{MSP}{modified set-partitioning}
\newacronym[]{ssp}{SSP}{semi set-partitioning}
\newacronym[]{fhd}{FHD}{free Hamming distance}
\newacronym[]{mfhd}{MFHD}{maximum free Hamming distance}
\newacronym[]{ods}{ODS}{optimal distance spectrum}
\newacronym[]{iud}{i.u.d.}{independent and uniformly distributed}
\newacronym[]{ud}{u.d.}{uniformly distributed}
\newacronym[]{iid}{i.i.d.}{independent, identically distributed}
\newacronym[]{ami}{AMI}{accumulated mutual information}
\newacronym[]{bico}{BICO}{binary-input continuous-output}
\newacronym[]{gh}{GH}{Gauss--Hermite}
\newacronym[]{gum}{GUM}{Gaussian--uniform mixture}

\newacronym[\glsshortpluralkey=BSs,\glslongpluralkey=base-stations]{bs}{BS}{base-station}
\newacronym[\glsshortpluralkey=MSs,\glslongpluralkey=mobile-stations]{ms}{MS}{mobile-stations}

\newacronym[]{phy}{PHY}{physical layer} 
\newacronym[]{rlc}{RLC}{Radio-Link control} 
\newacronym[]{ran}{RAN}{Radio Access Network} 
\newacronym[]{llc}{LLC}{logical link control} 
\newacronym[]{tcp}{TCP}{transmission control protocol} 
\newacronym[]{mac}{MAC}{media access control} 
\newacronym[]{fft}{FFT}{fast Fourier transform} 
\newacronym[]{ft}{FT}{Fourrier transform}
\newacronym[]{cf}{CF}{characteristic function} 
\newacronym[]{mgf}{MGF}{moment generating function} 
\newacronym[]{ee}{EE}{energy efficiency} 
\newacronym[]{eb}{EB}{energy per bit}
\newacronym[]{kkt}{KKT}{Karush--Kuhn--Tucker} 
\newacronym[]{mcs}{MCS}{modulation/coding scheme} 
\newacronym[]{fec}{FEC}{forward error correction}
\newacronym[]{arq}{ARQ}{automatic repeat request}
\newacronym[]{harq}{HARQ}{hybrid ARQ}
\newacronym[]{tarq}{TARQ}{truncated HARQ}
\newacronym[]{ir}{IR}{incremental redundancy}
\newacronym[]{rpr}{RR}{repetition redundancy}
\newacronym[]{rrharq}{RR-HARQ}{repetition redundancy HARQ}
\newacronym[]{irharq}{IR-HARQ}{incremental redundancy HARQ}
\newacronym[]{ack}{ACK}{positive acknowledgment}
\newacronym[]{nack}{NACK}{negative acknowledgment}
\newacronym[]{hol}{HoL}{head of the line}
\newacronym[]{crc}{CRC}{cyclic redundancy check}
\newacronym[]{dp}{DP}{dynamic programming}
\newacronym[]{gp}{GP}{geometric programming}
\newacronym[]{per}{PER}{packet error rate}
\newacronym[]{ber}{BER}{bit error rate}
\newacronym[]{op}{OP}{outage probability}
\newacronym[]{spa}{SPA}{saddle-point approximation}
\newacronym[]{mrc}{MRC}{maximum ratio combining}
\newacronym[]{mdp}{MDP}{Markov decision process}
\newacronym[]{lp}{LP}{linear programming}
\newacronym[]{pomdp}{POMDP}{partially observable Markov decision process}
\newacronym[]{psimdp}{PSI-MDP}{partial state information Markov decision process}
\newacronym[]{scpp}{SCPP}{stochastic shortest path problem}

\newacronym[]{forw}{frwd}{forward}
\newacronym[]{feed}{fdbk}{feedback}

\newacronym[]{mm}{MM-HARQ}{multi-message HARQ}
\newacronym[]{xp}{XP-HARQ}{cross-packet HARQ}
\newacronym[]{ts}{TS}{time-sharing}
\newacronym[]{sc}{SC}{superposition coding}
\newacronym[]{sbrq}{SBRQ}{systematic backward retransmission}
\newacronym[]{brq}{BRQ}{backward retransmission}
\newacronym[]{lharq}{L-HARQ}{layer-coded HARQ}
\newacronym[]{anlharq}{AoN-HARQ}{all-or-none L-HARQ}
\newacronym[]{vlharq}{VL-HARQ}{variable-length HARQ}

\newacronym[]{pp}{PPP}{point process}
\newacronym[]{ppp}{PPP}{Poisson point process}

\newacronym[]{fide}{FIDE}{F\'ed\'eration Internationale des \'Echecs}
\newacronym[]{fifa}{FIFA}{F\'ed\'eration Internationale de Football Association}
\newacronym[]{fivb}{FIVB}{F\'ed\'eration Internationale de Volleyball}
\newacronym[]{epl}{EPL}{English Premier League}
\newacronym[]{nhl}{NHL}{National Hockey League}
\newacronym[]{nfl}{NFL}{National Football League}
\newacronym[]{ipl}{IPL}{Indian Premier League}

\newacronym[]{sg}{SG}{stochastic gradient}
\newacronym[]{lms}{LMS}{least mean squares}
\newacronym[]{rls}{RLS}{recursive least squares}
\newacronym[]{vss}{VSS}{variable step-size}
\newacronym[]{hfa}{HFA}{home-field advantage}
\newacronym[]{ha}{HA}{home advantage}
\newacronym[]{mov}{MOV}{margin of victory}
\newacronym[]{ac}{AC}{Adjacent Categories}
\newacronym[]{cl}{CL}{Cumulative Link}
\newacronym[]{rps}{RPS}{Ranked Probability Score}
\newacronym[]{mse}{MSE}{Mean Squared Error}
\newacronym[]{mmse}{MMSE}{Minimum Mean Squared Error}
\newacronym[]{rmse}{RMSE}{Root Mean Squared Error}
\newacronym[]{map}{MAP}{maximum a posteriori}
\newacronym[]{ml}{ML}{maximum likelihood}
\newacronym[]{loo}{LOO}{leave-one-out}
\newacronym[]{alo}{ALO}{approximate leave-one-out}
\newacronym[]{logo}{LOGO}{leave-one-game-out}
\newacronym[]{alogo}{ALOGO}{approximate leave-one-game-out}
\newacronym[]{msd}{MSD}{mean-square deviation}
\newacronym[]{lop}{LOP}{linear ordering problem}

\newacronym[]{svd}{SVD}{singular values decomposition}

\newacronym[]{skf}{SKF}{Simplified Kalman Filter}
\newacronym[]{vskf}{vSKF}{\emph{vector-covariance} Simplified Kalman Filter}
\newacronym[]{sskf}{sSKF}{\emph{scalar-covariance} Simplified Kalman Filter}
\newacronym[]{fskf}{fSKF}{\emph{fixed-variance} Simplified Kalman Filter}
\newacronym[]{kf}{KF}{Kalman Filter}
\newacronym[]{gelo}{G-Elo}{Generalized Elo}

\newacronym[]{tpb}{TPB}{tensor-product-basis}

\newtheorem{lemma}{Lemma}

\newtheorem{example}{Example}

\begin{document}

\title{Why winning a soccer game is worth 5 points 
}
\author{Leszek Szczecinski
\thanks{%
L.~Szczecinski  is with Institut National de la Recherche Scientifique, Montreal, Canada. [e-mail: Leszek.Szczecinski@inrs.ca].}%
}%

\ifdefined\JSA
\maketitle    
\fi
\ifdefined\JQAS
\originalmaketitle  
\fi

\setstretch{1.6} 

\begin{abstract} 
Ranking is used in sport leagues to determine a champion and/or to decide on promotion/relegation of teams. Arguably, the best known ranking method relies on scores obtained by cumulating the points associated with the wins and the draws of all teams, which are then ranked by sorting the score obtained. There are two main problems with this ranking method. First, the meaning of the ranking is undefined, and, second, it depends on the relative value of the wins that is arbitrarily set. We remedy these issues by introducing a probabilistic model of the game results and by showing an interpretation of the ranking that is consistent with the model. We also propose a methodology to estimate the parameter of the model which allows us to objectively determine the value of the win. In particular, using data from the association football (soccer), we show that the value of the win is close to five (5) points.
\end{abstract}

\section{Introduction}\label{Sec:Intro}

The ranking of teams in sports leagues is a key element of sports competition. It allows us to determine the champion or to promote and relegate teams. In general, we expect it to map the observed results of the games into a single parameter that is somehow related to the ``strength'' of the team.

The problem of ranking has been tackled in the literature mainly from the perspective of the so-called power ranking, where teams are characterized by continuous-valued strengths which must be inferred from the results of the games, \eg in \cite{Maher82}, \cite{Fahrmeir92}, \cite{Karlis08}, \cite{Egidi18}, \cite{Ley19}, \cite{Lasek20}, \cite{Szczecinski22}. By sorting the estimated strengths, we obtain the ranking. This is the principle underlying the current ranking used by \gls{fifa} and \gls{fivb}.

Power ranking is usually implemented using adaptive algorithms, where strengths are updated after each game played, and many games must be observed before the strengths converge close to the optimal solution. For this reason, it is useful in competitions that do not have a time-limited calendar, \eg in international meetings where games are played continuously throughout the years. This is why the power ranking could be implemented by \gls{fifa} and \gls{fivb} to rank the national teams.

On the other hand, sports leagues have relatively short seasons in terms of the number of games played per team. The convergence is thus difficult to guarantee, and since the seasons are treated independently, the results of one cannot be used to initialize the other. To avoid these problems, the ranking method used in sports leagues consists in calculating a score for each team by cumulating the points associated with their results (such as wins and draws). The assignment of points to the game results is called a \emph{scoring rule}. The teams are then ranked by sorting the score obtained. 

We focus here on the very popular case of games in which teams can win, lose, or draw the game, typical in association football (soccer) games, and we address two main issues of the ranking method, which relies on cumulative score.

First, the ranking results are offered without a clear-cut interpretation;  this problem is not caused by the ranking method as such, but rather by the lack of explicit model underlying the method. Therefore, we propose a model which endows the ranking with a simple interpretation. A corollary result is that while the existing rankings may be suitable to identify the best teams (in the sense we will show), they do not identify the worst ones well. We propose a new ranking targeting the worst teams.

Second, the ranking (result) depends on the relative value of the wins, whose value seems entirely subjective and may vary between sports. For example, in today's association football games, the scoring rule (3-1-0) is used, that is, three points are given for each win, one point for a draw, and no points for the loss. Before 1981, the rule (2-1-0) was used and the official reason for increasing the win-value was to motivate the teams to win the game. Whether we agree or not with this rationale, the win value appears to be entirely subjective. The main goal of this work is to show that it can be objectively determined from the data.


This work is organized as follows. We define the problem and the relevant notation in \secref{Sec:Model} while the model is introduced in \secref{Sec:simplified.ranking} where we also show that summing the points of wins and draws is, in fact, a ranking method consistent with the model. In \secref{Sec:Inference} we show how to estimate the model parameters from the observations and which scoring rules match the data. We give examples related to practical cases in association football comparing the results in the seasons played under the different scoring rules. Conclusions and recommendations are presented in \secref{Sec:Conclusions}.

\section{Problem definition}\label{Sec:Model}
Consider $M$ teams indexed with integers from the set $\mcI=\set{1,\ld, M}$, endowed with an intrinsic order $\rho_1\prec\rho_2\prec\ld\prec\rho_{M}$, $\rho_m \in \mcI$, from which we understand that team $\rho_1$ is ordered before team $\rho_2$, the latter is ordered before $\rho_3$, etc. In the context of sport ranking, we interpret $m \prec n$ as ``$m$ is better than $n$''; consequently, $\rho_1$ is the best team, $\rho_2$ is the second best, and $\rho_M$ is the worst. But this is merely semantics: in social science, we interpret $m \prec n$ as ``$m$ is preferred to $n$'' \citep{Slater61}, \citep{Maas95}.

The order defined by a vector $\brho=[\rho_1,\ld, \rho_M]\in\Pi_{\mcI}$, where $\Pi_{\mcI}$ is a set of all permutations of $\mcI$, cannot be directly observed: we may only infer the order from some observable results affected by $\brho$. The result of the inference $\hat{\brho}\in\Pi_{\mcI}$ is called a ranking. 

The common problem is to find the ranking $\hat\brho$ through pairwise comparisons, which appear naturally in sports competitions \citep{Flueck75}, where teams are selected and compared against each other through one-on-one games; this is repeated for different pairs and, since the comparison results are not reliable, the same pairs are compared many times, which amounts to playing more games.

Here, we consider a very common situation, where the comparison between teams $m$ and $n$  yields three possible outcomes: $m$ wins against $n$, or $n$ wins against $m$, or else $m$ and $n$ draw. We gather these results in the $M \times M$ matrices $\matW$ and $\matD$, whose entries $w_{m,n}$ and $d_{m,n}$, are natural numbers that indicate, respectively, how many times $m$ won against $n$ and how many times $m$ drew playing against $n$. 

For convenience, we set $w_{m,m}=0$ and $d_{m,m}=0$; $T_{\tr{D}}=\sum_{m<n} d_{m,n}$ is used to denote the total number of draws, $w_m=\sum_{n=1}^M w_{m,n}$ -- the number of games in which $m$ won, $l_m=\sum_{n=1}^M w_{n,m}$ -- the number of games in which $m$ lost, and $d_m=\sum_{n=1}^M d_{m,n}$ -- the number of games in which $m$ drew. We assume that the number of games played by each team is equal to $T_0=w_m+l_m+d_m$ and the total number of games is equal to $T$.

Finding the order from the results stored in $\matW$ and $\matD$ can be done in many different ways. Arguably, the best known method is based on the calculation of a \emph{score} 
\begin{align}
\label{score.definition}
    s_m=w_m +\xi d_m,
\end{align}
where $\xi<1$ is the nominal ``value'' of the draw.

By sorting the score $s_m$ in descending order, the first element identifies the ``best'' team and this rule is so common that it often is confounded with a definition of what the best team means. However, such a definition would be circular: since the highest score determines which team is the best, we cannot define ``being the best'' through the score. Another conceptual difficulty in interpreting the ranking is due to the fact that the score depends on $\xi$ which is subjectively defined

The score $s_m$ is fractional and, for presentation purposes, may be more convenient to use
\begin{align}
    s'_m= \frac{1}{\xi}s_m = \kappa w_m + d_m,
\end{align}
where $\kappa=1/\xi$ is the nominal value of the win compared to the draw when the latter is worth one (point). From the ranking perspective, the scoring rules (1,$\xi$,0) and $(\kappa,1,0)$ are equivalent, but the latter is more common and $\kappa$ is often defined by integer values. Of course, there is no particular reason to believe $\kappa$ should be integer; it is just simpler to count. As we will see in \secref{Sec:Inference}, the $\xi$-based scoring rule is preferable from the point of view of estimation procedures.

\section{Model and practical rankings}\label{Sec:simplified.ranking}

To eliminate the interpretation difficulties of the ranking methods and to verify if the nominal values of $\xi$ (or $\kappa$) are related to the empirical results, we must first define a mathematical model that relates the observed results to the order between the teams. 

We use a probabilistic framework where the probability of the game result is conditioned on the underlying order. Namely,
\begin{align}
\label{Pr.alpha}
    \log \PR{m ~\tr{wins against}~ n| m\prec n} &= \alpha,\\
\label{Pr.beta}
    \log \PR{m ~\tr{loses against}~ n| m\prec n} &= \beta,\\
\label{Pr.gamma}
    \log \PR{m ~\tr{draws against}~ n} &=\gamma,
\end{align}
where the logarithms simplify manipulations on probabilities. 

We expect $\alpha>\gamma>\beta$ because we treat the game results as ordinal variables, where the probability of a better team losing should be smaller than the probability of drawing, and the latter -- smaller than the probability of winning. 

Thus, we have a simple probabilistic interpretation of the order $\brho$: the better team has more chances to win than to lose or to draw. More precisely, it wins with probability $\e^{\alpha}$. From the law of total probability, we must have $\e^{\alpha}+\e^{\gamma}+\e^{\beta}=1$.

This model has a long history. It extends the one proposed in \cite{Slater61} by including the draws, and is very close to the model which can be found in \citet{Singh68}, \citet{Tiwisina19}. The main difference from the latter is that we use the parameters $\alpha$, $\beta$, and $\gamma$ which do not depend on the teams $(m, n)$ being compared. This is a pragmatic simplification from which the current ranking methods can be derived.

With the model defined, we are ready to find the ranking, \ie the estimate of the order.

\textbf{Joint ranking}

The common formulation in the literature is to solve the \emph{joint} \gls{ml} estimation problem \citep{deCani69}, where all positions are identified simultaneously: 
\begin{align}
\label{joint.ML}
    \hat\brho&=\argmax_{\brho\in\Pi_{\mcI}} J(\brho)\\
\label{J.brho}
    J(\brho)&=\log \PR{\matW,\matD| \brho}=\log \PR{\matW| \brho}+\log \PR{\matD| \brho}\\
    \log \PR{\matW| \brho}
    &=\alpha C(\brho) + \beta (T-T_{\tr{D}}-C(\brho))\\
    \log \PR{\matD| \brho}
    &=\gamma T_{\tr{D}}
\end{align}
where $\Pi_{\mcI}$ is the space of all possible permutations of $\mcI$ and
\begin{align}
    C(\brho)&=\sum_{m=1}^M\sum_{n=m+1}^{M} w_{\rho_m,\rho_n}
\end{align}
is the consistency of the order $\brho$ which tells us how many times in the season the better team (as indicated by $\brho$) won the game. Of course, $C(\brho)$ depends on $\matW$.

Since we assume $\alpha>\beta$, the optimization in \eqref{joint.ML} is equivalent to the maximization of $C(\brho)$ resulting in a well-known \gls{lop}, for which many solutions have been proposed; see \eg \cite{Schiavinotto04}, \cite{Charon10}.

However, the \gls{lop} formulation is not quite suitable for ranking in sports. First, finding $\hat\brho$ in \eqref{joint.ML} requires non-trivial numerical procedures \citep{Schiavinotto04}, which lacks transparency. Second, there may be many equivalent solutions $\hat\brho$. In fact, \citet{Szczecinski22c} shows that there may be hundreds or even thousands of equivalent solutions to the ranking problem defined by \eqref{joint.ML} even for a relatively small value of $M\approx 13$.

Clearly, a ranking that produces opaque and/or ambiguous results is of little use in the context of sport competitions. The Bowl Championship Series (BCS) in college-level American football provide a cautious story, where the ranking which was based on non-transparent principles used to generate a lot of controversy, see \citet[Ch.2~]{Langeville12_book} for more context. To avoid such issues, practical (sport) rankings tend to have simple and explicitly defined rules.

Therefore, sorting the score $s_m$ is a pragmatic ranking solution. In fact, for win/loss games (no draws), it is a well-known but suboptimal solution of the \gls{lop} \citep{Chenery58}, \citep{Slater61}, \citep{Marti12}; for the win/draw/loss case it is known as the Copland method \citep{Saari96} where $\xi=\frac{1}{2}$ is used. This ranking method is simple and transparent; yet, since we do not know how it is derived, we cannot interpret it in unambiguous terms. More importantly, the ranking method is not explicitly related to any model, and without the latter, we do not know how to find $\xi$ objectively.

\subsection{Ranking: Probabilistic Interpretation}

We first want to dispel the confusion of ranking interpretation by giving it a clear probabilistic meaning under the model we introduced. The approach used by \citet{Ben-Naim07} was similar in spirit by interpreting $s_m$ as a monotonic function of the estimated probability of winning (against a randomly chosen opponent). Such a formulation was possible in the case of binary (win/loss) games. We take a more general approach applicable in the case of win/draw/loss, by addressing directly the issue of ordering. This is done using the following criterion.

\begin{lemma}\label{Lemma:best}
The logarithmic likelihood of the team $m$ being the best (\ie $\rho_1=m$), can be written as
\begin{align}
    \log\PR{\matW, \matD|\rho_1=m}=(\alpha-\beta)s_m + \tr{Const.},
\end{align}
where $s_m=w_m + \xi d_m$ is the score and the draw value is given by
\begin{align}
\label{xi}
    \xi &= \frac{\gamma-\beta}{\alpha-\beta}.
\end{align}
\end{lemma}
\begin{proof}
From the independence of the results counted in $\matW$ and $\matD$ we obtain the following:
\begin{align}
\nonumber
    \log\PR{\matW, \matD|\rho_1=m}
    &=
    \sum_{n} w_{m,n}\log \PR{m\tr{~wins against~}n|m\prec n}\\
\nonumber
    &+
    \sum_{n} w_{n,m}\log \PR{m\tr{~loses against~}n|m\prec n}\\
\label{Pr.rho1}
    &+
    \sum_{n} d_{m,n}\log \PR{m\tr{~draws against~}n|m\prec n}.
\end{align}
Since we assume that the team $m$ is the best, the condition $m\prec n$ is satisfied for all $n$. Therefore, using \eqref{Pr.alpha}-\eqref{Pr.gamma}, we may write \eqref{Pr.rho1} as
\begin{align}
\nonumber
    \log\PR{\matW, \matD|\rho_1=m}
    &=
    \alpha w_{m} + \beta (T_0-w_m-d_m) +\gamma d_{m}\\
\label{Pr.rho1.sm}
    &=(\alpha-\beta)s_m + \beta T_0,
\end{align}
where factoring out $(\alpha-\beta)$ is possible because $\alpha>\beta$ (a reasonable assumption confirmed by the estimation results).
\end{proof}

With the number of games $T_0$ being the same for all teams, the sorting of $s_m$ ranks the teams according to their likelihood of being the best. This is the interpretation we were looking for: the team with the highest score $s_m$ is not the ``best'' but rather has the highest likelihood of being the best. This simply emphasizes the fact that the order is not directly observable and can only be inferred from the competition results. 

This finding gives us some confidence that the model \eqref{Pr.alpha}-\eqref{Pr.gamma} may be seen as a mathematical foundation for current sport rankings. 

Another interesting implication is that the team with the lowest score $s_m$ is not the worst, but it is the least likely to be the best. Such an interpretation seems awkward, especially when being the worst team is meant to have real consequences (such as suffering relegation to the lower-tier league). Then, finding the last team by sorting $s_m$ misses the point, and a more natural question to ask is rather: which team is the most likely to be the worst? Using the same approach as in the proof of \lemref{Lemma:best} above, we obtain the following.

\begin{lemma}\label{Lemma:worst}
The log-likelihood of the team $m$ being the worst (\ie $\rho_M=m$), can be written as
\begin{align}
\label{Pr.rhoM}
    \log\PR{\matW, \matD|\rho_M=m}=(\alpha-\beta)\ov{s}_m + \tnr{Const.},
\end{align}
where 
\begin{align}
\label{ov.sm}
    \ov{s}_m & = l_m + \xi d_m
\end{align}
\end{lemma}

The form of \eqref{ov.sm} is identical to \eqref{score.definition} if we replace wins $w_m$ with losses $l_m$, and the largest value of $\ov{s}_m$ corresponds to most likely the worst team. Therefore, by sorting $\ov{s}_m$ in \emph{ascending} order, we obtain a new ranking, where the last team is most likely the worst one.

Note that the ranking based on $\ov{s}_m$ treats the draws as contributors to identify the worst teams, \ie larger $d_m$ increases the likelihood of being the worst. This may seem counter intuitive but only because we are used to looking at the ranking from a unique perspective of adding points to improve the ranking position. 

To dispel doubts and gain a better understanding of this issue, we can compare two teams with the same number of losses $l_m$. Quite naturally, between these teams, the one with the smaller number of wins $w_n$ is ``worse". But since we fixed the number of losses, this also means that having a larger $d_m$ identifies the worse team.

We now have two distinct rankings: the one based on $s_m$ determines which teams are most likely the best, and the one based on $\ov{s}_m$ allows us to determine which teams are most likely the worst. 

In general, $s_m$ and $\ov{s}_m$ produce different rankings, and making decisions using both may depend greatly on the context. For example, assume that we must split the league into two groups, \eg when identifying the composition of the playoffs after the regular season. This split into two groups will, in general, not be the same if we use $s_m$ or $\ov{s}_m$. If both groups are of comparable size (\eg in the case of half-half split), it is not obvious what to do. Should we use $s_m$ to identify the best teams, or $\ov{s}_m$ -- to identify the worst ones? We have no clear solution to that conundrum and, most likely, the convention approach of using $s_m$ is more acceptable.

On the other hand, if we want to identify small subgroups of the best and the worst teams, both $s_m$ and $\ov{s}_m$ can be used; for example, to determine which teams should be promoted to the higher-tier league, we can use $s_m$ and, to determine which teams should be relegated to the lower-tier league, we can use $\ov{s}_m$.

There is, however, one case where $s_m$ and $\ov{s}_m$ produce the same rankings. We can easily note from \eqref{ov.sm} and \eqref{ov.sm} that 
\begin{align}
\label{}
    -\ov{s}_{m,\xi} & = s_{m,1-\xi} - T_0,
\end{align}
where we make explicit the dependence of $s_m$ and $\ov{s}_m$ on $\xi$ (by subindexing). Since the constants (here $T_0$) are irrelevant to the sorting, we know that the sorting of $-\ov{s}_{m,\xi}$ and $s_{m,1-\xi}$ (in descending order) yields the same results. Therefore, ranking based on $s_{m,\xi}$ and on $\ov{s}_{m,\xi}$, yield identical results if $\xi=0.5$, that is, in the pre-1981 scoring rule (1-0.5-0) or, equivalently (2-1-0).

\begin{example}\label{Ex:sm_vs.ovsm}
To illustrate how $s_m$ and $\ov{s}_m$ affect the rankings, \tabref{tab:positions} provides the rankings in the 2002/03 seasons of \gls{epl} for a nominal value $\kappa=5$ (or $\xi=\frac{1}{5})$ using both $s_m$ and $\ov{s}_m$; where the reference is the current ranking based on $\kappa=3$ and $s_m$.

It should be emphasized that we evaluate the ranking changes for \emph{given} results. On the other hand, to verify whether changing the nominal value used in the ranking alters the performance of the teams, we must use data from seasons with different nominal values of $\xi$. We will do it in \exref{Ex:Inference}.

\begin{table}[]
    \centering
    \begin{tabular}{c|c||c|c||c|c}
    \multicolumn{2}{c||}{$\kappa=3$}
    &
    \multicolumn{4}{c}{$\kappa=5$}\\
    & $s'_m$ & & $s'_m$ & & $\ov{s}'_m$\\
    \hline
Man United [25-8-5] & 83 & Man United & 133 & Man United & 33 \\
Arsenal [23-9-6] & 78 & Arsenal & 124 & Arsenal & 39 \\
Newcastle [21-6-11] & 69 & Newcastle & 111 & Chelsea & 55 \\
Chelsea [19-10-9] & 67 & Chelsea & 105 & Liverpool & 60 \\
Liverpool [18-10-10] & 64 & Liverpool & 100 & Newcastle & 61 \\
Blackburn [16-12-10] & 60 & Everton & 93 & Blackburn & 62 \\
Everton [17-8-13] & 59 & Blackburn & 92 & Southampton & 73 \\
Southampton [13-13-12] & 52 & Man City & 81 & Everton & 73 \\
Man City [15-6-17] & 51 & Southampton & 78 & Bolton & 84 \\
Tottenham [14-8-16] & 50 & Tottenham & 78 & Middlesbrough & 85 \\
Middlesbrough [13-10-15] & 49 & Charlton & 77 & Tottenham & 88 \\
Charlton [14-7-17] & 49 & Middlesbrough & 75 & Fulham & 89 \\
Birmingham [13-9-16] & 48 & Leeds & 75 & Birmingham & 89 \\
Fulham [13-9-16] & 48 & Birmingham & 74 & Man City & 91 \\
Leeds [14-5-19] & 47 & Fulham & 74 & West Ham & 92 \\
Aston Villa [12-9-17] & 45 & Aston Villa & 69 & Charlton & 92 \\
Bolton [10-14-14] & 44 & Bolton & 64 & Aston Villa & 94 \\
West Ham [10-12-16] & 42 & West Ham & 62 & Leeds & 100 \\
West Brom [6-8-24] & 26 & West Brom & 38 & West Brom & 128 \\
Sunderland [4-7-27] & 19 & Sunderland & 27 & Sunderland & 142
    \end{tabular}
    \caption{Ranking in 2002/03 season of the \gls{epl}. The first column corresponds to the currently used ranking ($\kappa=3$), where the teams' records [win-draw-loss] are shown in brackets; the remaining rankings use $\kappa=5$. We rely on sorting of $s'_m$ (in descending order) and $\ov{s}'_m$ (in ascending order). Names of two teams are shown in bold, when the position change may be relevant: Everton advances to the sixth position (open a possibility of participation in subsequent tournament), and, using $\ov{s}_m$, Leeds is demoted to the third-bottom position which, in the current ranking, leads to relegation.}
    \label{tab:positions}
\end{table}

Sorting $s_m$, we aim to identify the best teams, and, using $\kappa=5$, there is a potentially important change, when Everton advances to the sixth position, which may lead to a post-season tournaments.\footnote{Only four top positions guarantee the participation in European tournaments; the fifth and the sixth ones help, if other, more involved conditions are satisfied.}

If we are interested in finding the worst teams, we use $\ov{s}_m$ and then, as we said, the losses are important, so Sunderland [4-7-27] and West Brom [6-8-24], whose records are dominated by losses, do not change their bottom positions. On the other hand, the third and fourth positions occupied, respectively, by West Ham [10-12-16] and Bolton [10-14-14], are altered because both teams, despite displaying more losses than wins, have an ambiguous record moderated by a very large number of draws. In contrast, Leeds' record [14-5-19] puts it in the third position with respect to the number of losses, and this fact demotes it from a comfortable bottom sixth position in the original ranking to the bottom third position (which would imply a relegation) in the ranking based on $\ov{s}_m$. In fact, it is easy to see that using $\ov{s}_m$, Leeds would be the third lowest, even for $\kappa=3$.

To go beyond a particular case of one season, we rank the \gls{epl} teams in 23 seasons 1995/96 -- 2018/19. Using the current ranking as a reference ($s_m$ and $\kappa=3$), we evaluate how many times in 23 seasons, varying $\kappa$ and/or using $\ov{s}_m$, the most relevant positions are altered. We consider the top position (being a champion), the top six positions (which may allow the teams to progress to another stage of the international competitions), and the three bottom positions (when the teams are relegated to the lower-tier league). We also consider changes in the three bottom positions assuming that the ranking is based on $\ov{s}_m$. The results are shown in \tabref{tab:epl.changes} 

\begin{table}[]
    \centering
     \begin{tabular}{c| c|c|c|c|c}
         &  \multicolumn{5}{c}{$\kappa$}\\
    ranking positions     & 2 & 3 & 4 & 5 & 6\\
         \hline
    champion ($s_m$) &
    1 & - & 0 & 0 & 0 \\
    top  6 ($s_m$)&
    2 & - & 3 & 5 & 6\\
    bottom 3 ($s_m$)&
    3 & - & 5 & 5 & 9 \\
    bottom 3 ($\ov{s}_m$) &
    3 & 10 & 10 & 12 & 14 
    \end{tabular}
    \caption{The seasons 1995/96 - 2018/2019 of the \gls{epl}: number of times the teams would be removed from their champion, top, or bottom positions when using values of $\kappa\neq 3$. The changes in bottom positions are also considered using the score $\ov{s}_m$ instead of $s_m$. Since $\kappa=3$ is used as a references for sorting of $s_m$, the respective values in the table are indicated with ``-".}
    \label{tab:epl.changes}
\end{table}

As we can see, only once the champion would be different due to the change from $\kappa=2$ (value before 1981) to $\kappa=3$. It happened in 2018, when Manchester City became champion with $w_m=32$ wins and $d_m=2$ draws, outranking Liverpool who had $w_m=30$ wins and $d_m=7$ draws; of course, the situation would reverse for $\kappa=2$. 

Note that using $\kappa>3$, the champion position is not challenged, but other relevant top positions are altered roughly once every four years. However, most of the changes occur in the bottom spots, and using $\ov{s}_m$ instead of $s_m$ significantly increases their number: changes would be observed every second season.

\end{example}

The results in \tabref{tab:positions} indicate that, for sufficiently large $\kappa$, when comparing $s'_m=\kappa w_m + d_m$, the number of draws $d_m$ is irrelevant as long as $w_m$ are different (because to compensate for one win, the team must have at least $\kappa$ more draws). In fact, we can see that, using $\kappa=5$, we essentially rank the teams according to their number of wins, $w_m$ (there are two exceptions where the teams with fewer wins are ranked before those with more: Southampton is ranked before Tottenham, Charlton and Leeds, and Middlesbrough is ranked before Leeds). Therefore, instead of $s_m$, we might use a two-criteria ranking: to identify the best teams, we may use the number of wins $w_n$ as the primary criterion, and, in the case where ties are observed (that is, the teams have the same values of $w_m$), we break them using the draws $d_m$ where a larger $d_m$ would correspond to a higher ranking. This concept is of course not new, and the current rankings are already based on multi-criteria evaluation, where the ties are broken by a goal-differential, a number of goals scored, etc.

Similarly, to identify the worst teams through $\ov{s}_m$, we can sort the number of losses $l_m$ (with the largest number of losses corresponding to the lowest ranking), and the ties would be broken again by draws $d_m$ (the larger $d_m$ would correspond to a lower ranking).

In fact, such a two-criteria ranking would eliminate the need to know $\kappa$. However, to do that, we should know if $\kappa$ \emph{is} large (or $\xi$ is small) and, given the model, we should be able to estimate $\kappa$ or $\xi$ from the data. This is what we do in the following.

\section{Estimation of the parameters of the model}\label{Sec:Inference}

Inference is made noting that under the model \eqref{Pr.alpha}-\eqref{Pr.gamma} the observations $\matW$ and $\matD$ (capturing the results in a particular season) depend on $\alpha$, $\beta$, and $\gamma$ and we start by obtaining their estimates. Due to the relationship $\e^\alpha +\e^\beta +\e^\gamma =1$, we only need to estimate two arbitrarily chosen parameters (say, $\alpha$ and $\gamma$). We start by calculating their likelihood as:
\begin{align}
\label{likelihood.WD}
    \Pr(\matW,\matD|\alpha,\gamma)
    =\sum_{\brho \in \Pi_{\mcI}} \PR{\matW,\matD|\alpha,\gamma,\brho}\PR{\brho}
    \propto\sum_{\brho \in \Pi_{\mcI}} \e^{J(\brho,\alpha,\gamma)},
\end{align}
where the marginalization of the orders $\brho$ assumes their uniform distribution over $\Pi_{\mcI}$ (no prior knowledge of the order), and the dependence on $\alpha$ and $\gamma$ is made explicit in $J(\brho,\alpha, \gamma)$ which we have already defined in \eqref{J.brho}.

The explicit enumeration of all rankings $\brho$ in \eqref{likelihood.WD}, with complexity $O(M!)$, is not feasible even for a moderate value of $M$ (\eg $M\approx 20$, which is a common number of teams in sport leagues). Therefore, we adopt the strategy proposed in \citet{Szczecinski22c}: noting that $J(\brho,\alpha,\gamma)=\alpha C(\brho) +\beta (T-T_{\tr{D}}-c(\brho))+\gamma T_{\tr{D}}$ depends on $C(\brho)$, and that the latter is discrete, \ie $C(\brho)\in\set{0,\ld, T}$, we can rewrite \eqref{likelihood.WD} as follows:
\begin{align}
    \Pr(\matW,\matD|\alpha,\gamma)
    &\propto
    \sum_{t=0}^T c_t \e^{\alpha t + \beta (T-T_{\tr{D}} - t) + \gamma T_{\tr{D}}},
\end{align}
where $\set{c_t}_{t=0}^T$ is the Slater spectrum of the matrix $\matW$ \citep{Szczecinski22c} whose coefficients $c_t$ indicate the number of orders $\brho\in\Pi_{\mcI}$ that have the same $C(\brho)=t$. It turns out that the Slater spectrum can be calculated recursively with complexity $O(M^2 2^M)$ \citep{Szczecinski22}; for moderate $M\approx 20$, the calculation can be done very easily on average desktop computers.

If we want to estimate the parameters $\alpha$ and $\gamma$ from an ensemble of seasons, we gather the results in sets of matrices $\un{\matW}=\set{\matW_1,\ld,\matW_K}$ and $\un{\matD}=\set{\matD_1,\ld,\matD_K}$, where $\matW_k$ and $\matD_k$, are the win and draw matrices indexed with the season $k$, where $K$ is the total number of seasons in the data set. We assume that the matrices $\matW_k$ and $\matD_k$, conditioned on the order $\brho_k$, are generated according to the same probabilistic model (defined by $\alpha$ and $\gamma$), and that the orders $\brho_k$ are independent between seasons. Then the likelihood of the ensemble of results is given by 
\begin{align}
\label{final.likelihood}
    \PR{\un{\matW}, \un{\matD} | \alpha,\gamma} 
    & = 
    \prod_{k=1}^K \PR{\matW_k, \matD_k|\alpha,\gamma}.
\end{align}

It may be debatable whether multiple comparisons of the \emph{different} teams may share the same parameter ($\alpha$ and $\gamma$) or whether $\brho_k$ are independent (since they can represent the same teams in different seasons). Although these assumptions seem well justified in sports, where comparisons are made according to the same rules and teams' composition changes from season to season, to avoid the debate, we note that, since the same ranking methods are applied across seasons and leagues, the assumption of independence among $\brho_k, k=1,\ld,K$ is implicit in current ranking methods.

Of course, the ``ensemble" of seasons may be defined as we wish. It may mean many seasons of the same league or seasons of different leagues with different values of $M$, or simply one season.

At this point, we can parameterize the model using directly the interpretation of $\alpha$, $\beta$ and $\gamma$ as probabilities, \ie, $q_{\alpha}=e^{\alpha}$, $q_{\beta}=e^{\beta}=1-q_\alpha-q_\gamma$, and $q_{\gamma}=e^{\gamma}$. 

The posterior distribution of $q_\alpha$ and $q_\gamma$ is then given by
\begin{align}
\label{posterior.qq}
    \pdf(q_{\alpha}, q_{\gamma} | \un{\matW}, \un{\matD})
    &\propto
    \PR{\un\matW,\un\matD|\log q_\alpha,\log q_\gamma} \pdf(q_\alpha,q_\gamma),
\end{align}
where the prior $\pdf(q_\alpha,q_\gamma)$ should include knowledge of the constraints, \ie $\pdf(q_\alpha,q_\gamma)\propto \IND{q_\alpha>q_\gamma}$. In plain words, we assume that $q_\alpha$ and $q_\gamma$ are uniformly distributed over the space limited by the constraints $q_\alpha+q_\gamma\leq 1$ (law of total probability) and $q_\alpha>q_\gamma$ (prior $\alpha>\gamma$). On the other hand, to assess how the solution behaves, we do not impose the constraint $\gamma>\beta$. This allows us to obtain negative values of $\hat\xi$.

The estimate of $\xi$ can be then found, for example, as a posterior mean:
\begin{align}
\label{Ex.xi}
    \hat\xi&=\Ex_{q_{\alpha},q_{\gamma}} 
    \big[\xi(q_{\alpha},q_{\gamma})\big],
\end{align}
where 
\begin{align}
    \label{xi.qq}
    \xi(q_{\alpha},q_{\gamma})&=\frac{\log(q_{\gamma}/q_{\beta})}{\log(q_{\alpha}/q_{\beta})}
\end{align}
returns the value of $\xi$ defined in \eqref{xi}, while 
the uncertainty of the estimation will be characterized by the standard deviation
\begin{align}
\label{xi.std}
    \sigma_\xi&=\sqrt{\Ex\big[\xi^2(q_{\alpha},q_{\gamma})\big]- {\hat\xi}^2}.
\end{align}

The posterior means of the parameters $q_\alpha$, $q_\beta$, $q_\gamma$ can also be obtained in the same manner as
$\hat{q}_{\alpha} = \Ex \left[q_{\alpha}\right]$, $\hat{q}_{\beta} = \Ex \left[q_{\beta}\right]$, and $\hat{q}_{\gamma} = \Ex \left[q_{\gamma}\right]$.

The integration over the distribution \eqref{posterior.qq} required to obtain the expectation $\Ex_{q_\alpha, q_\gamma}[\cd]$ was performed numerically using the trapezoidal rule.

\begin{example}[Inference Results]
\label{Ex:Inference}
Let us consider soccer games in the most important professional European association football leagues.  We use the results of the \gls{epl}, the English Football League Championship (EFL Ch.), the English Football League One (EFL 1), and the English Football League Two (EFL 2); we consider 24 back-to-back seasons 1995/96 -- 2018/2019 chosen so that the number of teams is constant per league ($M=20$ for EPL and $M=24$ for EFL Ch., EFL 1, and EFL 2) and we avoid the pandemic-affected seasons. We also use the results from two top-tier German leagues (Bundesliga's first and second divisions, with $M=18$, seasons 1995/96 -- 2018/2019) and the Spanish ones (La Liga's Primera Division, $M=20$ and Segunda Division, $M=22$, seasons 1997/98 -- 2018/2019). The game results are in the public domain and were obtained from the repository \cite{football-data}.

We next considered the games played when the ranking was based on the scoring rule (2-1-0): these are pre-1981 seasons in the \gls{epl} and pre-1994 seasons in the Bundesliga (except for the corruption-affected season 1971/72 and the season 1991/92 with $M=20$). Data is available at \cite{worldfootball.net}

The estimation results are shown in Table~\ref{tab:kappa}, where we show the estimated probabilities of the win, loss, and draw, $\hat{q}_{\alpha}$, $\hat{q}_{\beta}$, and $\hat{q}_{\gamma}$. The estimated value of the draw value $\hat\xi$ is shown, and the reliability of the estimate may be assessed through the credible interval $(\hat\xi-2\sigma_{\xi}, \hat\xi+2\sigma_{\xi})$.

As we can see, only the results in the Spanish La Liga Segunda Division yield $\hat\kappa=3.2$ which is quite close to the value currently used, $\kappa=3$. The estimation that takes into account the ensemble of all seasons in all leagues yields a draw-value close to $\hat\xi=0.25$. This implies that the scoring rule (5-1-0), where five points are given for winning the game, could match the empirical data. This result explains the title of this work.

On the other hand, the results are quite different in seasons based on the nominal draw-value $\xi=0.5$, \ie for the scoring rule (2-1-0). In fact the estimated draw-value is then close to zero, \ie the empirically justified scoring rule was (1-0-0).

\end{example}

\begin{table}[]
    \centering
    \begin{tabular}{l|c|c|c||c||c}
    League  & 
    $\hat{q}_\alpha$ & $\hat{q}_\gamma$ & $\hat{q}_\beta$ & $\hat{\xi}\pm 2\sigma_{\xi}$ & $\hat{\kappa}=1/\hat{\xi}$\\
    \hline
     EPL    & 0.55 & 0.26 & 0.20  & 0.25 $\pm$ 0.06 & 3.9\\
     EFL Ch.& 0.49 & 0.28 & 0.24  & 0.23 $\pm$ 0.06 & 4.3\\
     EFL 1  & 0.49 & 0.27 & 0.24  & 0.18 $\pm$ 0.07 & 5.6\\
     EFL 2  & 0.48 & 0.28 & 0.25  & 0.17 $\pm$ 0.07 & 5.9\\
Bundesliga  & 0.52 & 0.25 & 0.22  & 0.16 $\pm$ 0.08 & 6.4\\
Bundesliga 2 & 0.48 & 0.28 & 0.24  & 0.18 $\pm$ 0.10 & 5.5\\
La Liga     & 0.53 & 0.25 & 0.22  & 0.13 $\pm$ 0.07 & 7.5\\
La Liga 2   & 0.45 & 0.30 & 0.25  & 0.31 $\pm$ 0.09 & 3.2\\
\hline
Ensemble    & \textbf{0.49} & \textbf{0.27} & \textbf{0.24} & \textbf{0.21 $\pm$ 0.03} & \textbf{4.9}  \\
\multicolumn{6}{}{}\\
        EPL ($<$1981) & 0.49 & 0.25 & 0.26 & -0.04 $\pm$ 0.08 \\
Bundesliga ($<$1994)  & 0.50 & 0.26 & 0.25 & 0.05 $\pm$ 0.09 \\
\cline{1-5}
Ensemble    & \textbf{0.49} & \textbf{0.255} & \textbf{0.255} & \textbf{0.00 $\pm$ 0.06} &  \\

\end{tabular}
        \caption{Estimates of the model parameters: $\hat{q}_\alpha$ is the estimated probability that the better team wins, $\hat{q}_\beta$ is the probability it loses, and $\hat{q}_\gamma$ -- the probability it draws. The estimated draw-value $\hat\xi$ is shown together with double standard deviation expressing uncertainty; the estimate of the win-value  $\hat\kappa$ is obtained by inverting $\hat{\xi}$ (but not shown when $\hat\xi\approx 0$, \ie in the case of pre-1981 games of the \gls{epl} and pre-1994 games of Bundesliga).}
    \label{tab:kappa}
\end{table}

\section{Conclusions}\label{Sec:Conclusions}
The probabilistic model we introduced, together with the estimation rule we proposed, give an unambiguous interpretation of the current ranking methods based on the cumulative score; it also leads to a new ranking, suitable for finding the worst teams. 

More importantly, the model provides a basis for statistical inference of the draw value $\xi$, as we did using data from association football (soccer) but, the proposed estimation methodology can be applied in win/draw/loss games as long as we are able to implement the numerical routines to find the Slater spectrum of the matrix $\matW$ defining the wins in the season (\ie for ``moderate'' value of $M$).

\textbf{Simplified rules}

If the win-value $\kappa$ is ``large'', the ranking strategies can be significantly simplified. In particular, to identify the (most likely the) best teams, the number of wins $w_n$ is sufficient, and the number of draws becomes a secondary ranking criterion: when the teams have the same number of wins (a tie), the larger number of draws indicates the better team. Similarly, to identify the worst teams, we rank them according to their number of losses $l_m$ (in ascending order), and the last teams are the (most likely the) worst. The tie (equal number of losses) is then broken again using draws.

Of course, if the latter recommendation is to be implemented, some care must be exercised to avoid the problems of overlapping results obtained from rankings based on $w_m$ and $l_m$. This may be done, for example, prioritizing the ``legacy" ranking in cases of problems that are nevertheless unlikely to occur if merely a few best and worst teams are to be identified.

\textbf{Results and recommendations in association football}

The results indicate that the nominal scoring rules in football do not match the empirical data. In particular, in the period under the nominal scoring rule (2-1-0), the data suggested that the rule should be (1-0-0), which essentially means that the conditional probabilities of wins and draws were the same and thus, the value of the draw was very small. On the other hand, under the current rule (3-1-0), the empirical observations suggest the rule (5-1-0).

Thus, it seems that the governing body, by introducing the scoring rule (3-1-0), \ie increasing the nominal value of the win to $\kappa=3$ (or equivalently reducing the nominal draw-value from $\xi=1/2$ to $\xi=1/3$), managed to change the pattern of the results: the (conditional) probability of wins is now smaller than the probability of draws. This is a constructive observation that suggests that we may adjust the ranking methods to match the reality of the game. It should always be the purpose of the ranking methods, but now the quality of the adjustment can be objectively assessed by looking at the difference between the nominal draw-value $\xi$ and its empirical counterpart $\hat\xi$. 

The results also suggest that finding the scoring rule cannot be done merely by observing the games under a particular scoring rule because the latter affects the results. Rather, we should go for an iterative process: we adjust the scoring rule, and next, we verify if the results match it. In case of a significant mismatch, we re-adjust the scoring rule, etc. 

In the case of football, although the empirical scoring rule is (5-1-0), instead of immediately imposing it as a nominal one, it is more prudent to make slow changes. For example, the rule may be first set to (4-1-0) and, observing the results for sufficiently long time (a decade or so), we may reassess the scoring rule: if $\hat{\kappa}\approx 4$ no adjustment will be needed.

\newcommand{\CFilesBib}{Common.Files.Bib}
\ifdefined\ARXIV

\else
\ifdefined\JSA
\bibliographystyle{apacite}  
\fi
\ifdefined\JQAS
\bibliographystyle{abbrvnat}  
\fi

\bibliography{\CFilesBib/references_rank,\CFilesBib/IEEEabrv,\CFilesBib/references_all}
\fi

\end{document}